\documentclass[twocolumn]{autart}    

\usepackage{graphicx}          

\usepackage{amsmath}
\usepackage{amsfonts}
\usepackage{color}
\usepackage{subcaption}

\newtheorem{theorem}{Theorem}
\newtheorem{corollary}{Corollary}
\newtheorem{lemma}{Lemma}

\newtheorem{proposition}{Proposition}

\theoremstyle{remark}
\newtheorem{remark}{Remark}

\theoremstyle{definition}
\newtheorem{definition}{Definition}

\newenvironment{problem}{
  \par\noindent\textbf{Problem.}\ }%
  {\par}

\newenvironment{proof}[1][Proof]{%
  \par\noindent\textbf{#1.}%
\ }{\hfill\rule{1.5mm}{1.5mm}\par}

\newcommand{\R}{\mathbb{R}}
\newcommand{\U}{\mathcal{U}}
\newcommand{\Lip}{\mathcal{L}}

\newcommand{\Kie}{\mathcal{K}_\infty^e}
\newcommand{\Ki}{\mathcal{K}_\infty}
\newcommand{\K}{\mathcal{K}}
\newcommand{\KL}{\mathcal{KL}}
\renewcommand{\t}[1]{\tilde{#1}}

\DeclareMathOperator*{\argmin}{argmin}

\begin{document}

\begin{frontmatter}

\title{Safety Under State Uncertainty: 
\\Robustifying Control Barrier Functions\thanksref{footnoteinfo}} 

\thanks[footnoteinfo]{
This work was partially supported by TII under project \#A6847 and by NSF award 2211146.
}

\author[UCLA]{Rahal Nanayakkara}\ead{rahaln@ucla.edu},    
\author[Caltech]{Aaron D. Ames}\ead{ames@caltech.edu},               
\author[UCLA]{Paulo Tabuada}\ead{tabuada@ee.ucla.edu}  

\address[UCLA]{Electrical and Computer Engineering Department, University of California at Los Angeles, USA}  
\address[Caltech]{Department of Mechanical and Civil Engineering, California Institute of Technology, USA}             


\begin{abstract}                          
Safety-critical control is a crucial aspect of modern systems, and Control Barrier Functions (CBFs) have gained popularity as the framework of choice for ensuring safety. However, implementing a CBF requires exact knowledge of the true state, a requirement that is often violated in real-world applications where only noisy or estimated state information is available. 
This paper introduces the notion of Robust Control Barrier Functions (R-CBF) for ensuring safety under such state uncertainty. 
Crucially, this framework does not require knowledge of the magnitude of uncertainty for the synthesis of a robust safe controller.
We formally characterize the class of robustifying terms that ensure robust closed-loop safety and show how a robustly safe controller can be constructed. We demonstrate the effectiveness of this approach through simulations and compare it to existing methods, highlighting the additional robustness and convergence guarantees it provides.

\end{abstract}

\end{frontmatter}

\section{Introduction}

Safety is a fundamental consideration in modern control system design. Formally, ensuring safety entails restricting the trajectories of a dynamical system to a predefined safe set.
In recent years, Control Barrier Functions (CBFs), originally introduced in \cite{cbf_main}, have emerged as the framework of choice for enforcing safety across a broad spectrum of applications including robotic systems \cite{cbf_robot}, autonomous vehicles \cite{cbf_aut_veh}, and satellites \cite{cbf_satellite}.
The CBF methodology requires that a controller satisfies a certain inequality to guarantee that the trajectories of the closed-loop system remain within the safe set.

While standard CBFs exhibit inherent robustness to small model perturbations, as shown in \cite{cbf_robustness}, various extensions have been proposed to address broader classes of uncertainties that frequently arise in practice.
These methods either seek to enforce safety of the original set for small amounts of uncertainty \cite{jankovic2018robust} or enforce safety on an inflated version of the safe set \cite{issf,tissf}, where the inflation is dependent on the amount of uncertainty. Typically, all these results are achieved by strengthening the regular CBF\footnote{We use the expression \emph{regular CBF} to refer to the original CBF condition introduced in equation (25) of \cite{cbf_main}.} inequality with additional robustifying terms, where the nature of the additional term is dependent on the type of uncertainty being addressed. For instance, \cite{ersin_dist_obs,xu2023disturbance} deals with unmodeled dynamics, while \cite{unmodeled_dynamics,sector_uncert} deals with input disturbances. A summary of some such common robustifying modifications can be found in \cite{parameterized2023ames}.

However, most of these modifications assume availability of the true state when computing the control input and verifying the modified CBF inequality; relatively less work has been done to address the uncertainty arising from using a state estimate for control. In practice though, this becomes a crucial aspect that must be taken into consideration since most implementations of control systems depend on observers or noisy measurements of the state, resulting in a state estimate which will often differ from the true state.
In this paper, we revisit the problem of constructing a CBF inequality providing safety guarantees that are robust with respect to uncertainty in the state estimate.
The approaches taken by previous work to address this problem can broadly be 
categorized into two approaches. 

One common approach involves designing specialized observers and leveraging their properties to enforce safety. 
For instance, \cite{panagou2022safe} shows that safety constraints can be enforced when an input-to-state stable observer is available, or, under more restrictive assumptions on the system dynamics, when a general “bounded-error” observer is used.

In \cite{xu2022observerCBF}, state uncertainty is addressed using function approximation techniques and an adaptive law. However, the approach is limited to a finite time horizon $T$, requiring either a large $T$ (which results in a large number of parameters to be estimated) or periodic resets.

A second line of work focuses on strengthening the standard CBF inequality with additional terms, in line with techniques from the robust CBF literature. In \cite{mrcbf2021guaranteeing,mrcbf2_iros}, the authors introduce the notion of a Measurement-Robust Control Barrier Function (MR-CBF), which guarantees safety under a known bound on state uncertainty. While this method ensures forward invariance of the original safe set when the bound is respected, it lacks the robustness and asymptotic convergence guarantees of standard CBFs, as we demonstrate in this paper. Moreover, an MR-CBF may not exist even when a regular CBF does.

The approach in this paper shares the following characteristics with the work in \cite{mrcbf2021guaranteeing}: 1) we strengthen the regular CBF inequality; 2) no assumptions are made on how the state estimate is obtained, i.e., no specific classes of observers are required. We refer to the novel CBF inequality proposed in this paper as a Robust Control Barrier Function (R-CBF). 
However, in contrast with the work in \cite{mrcbf2021guaranteeing}, our method possesses several advantages: 1) an R-CBF always exists whenever a CBF exists, provided there are no input constraints; 2) it always guarantees boundedness of solutions independently of the level of uncertainty; 3) neither constructing an R-CBF nor computing the resulting controller requires knowledge of a bound on the uncertainty.
We show that an R-CBF ensures that the original safe set remains invariant under small uncertainties, while an inflated safe set is rendered invariant under larger ones. 
Notably, our framework only requires a bound on the uncertainty to compute the exact safety guarantees for the resultant closed-loop system, namely determining whether the original set is invariant or calculating the precise size of the inflated safe set.

The main contributions of this paper are to axiomatically identify a class of robustifying terms leading to desired robustness properties, and to prove the resulting guarantees on the closed-loop system. 
As an intermediate step in obtaining our final result, 
we start by analyzing robustness with respect to actuation errors,
yielding stronger guarantees than what is available in the existing literature for this type of error.

The structure of this paper is as follows. In Section \ref{sct:prelims} we introduce the notation to be used throughout the paper and provide a brief overview of the CBF approach for enforcing safety. Next, in Section \ref{sct:main_result} we present our main result along with a few extensions/modifications and illustrate the use of an R-CBF in a practical system.
Finally, in Section \ref{sct:comparison_to_mrcbfs} we compare our approach to the MR-CBF approach and highlight the clear advantages that it possesses in terms of robustness properties.

\section{Preliminaries and Background} \label{sct:prelims}

\subsection{Notation}

Let $\R_{>0}$ and $\R_{\geq 0}$ denote the set of positive and non-negative real numbers, respectively. 
We denote by $\Vert \cdot \Vert$ the Euclidean norm and by \mbox{$B_\delta(y) = \{ x \in \R^n : \Vert x - y \Vert \leq \delta\}$}, the closed ball of radius $\delta$ centered at $y$ in $\R^n$.
The boundary of a set $C$ will be denoted by $\partial C$. For any $f:\R^n \to \R^n$, $g : \R^n \to \R^{n \times m}$ and differentiable $h:\R^n \to \R$ we denote by $L_fh(x) = \nabla_xh(x) \cdot f(x)$ and $L_gh(x) = \nabla_xh(x) \cdot g(x)$ the Lie derivatives of $h$ with respect to $f$ and $g$ respectively.

We call a continuous function $\alpha:\R_{\geq0} \to \R_{\geq0}$ a class $\K$ function (i.e., $\alpha \in \K$) if it is strictly monotonically increasing and satisfies $\alpha(0)=0$. We say $\alpha:\R \to \R$ is an extended class $\Ki$ function (i.e., $\alpha \in \Kie$) if it is strictly monotonically increasing and satisfies $\alpha(0)=0$, $\lim_{r \to -\infty}\alpha(r) =-\infty$ and $\lim_{r \to \infty}\alpha(r) =\infty$.
We call a continuous function $\beta : \R_{\geq0} \times \R_{\geq0} \to \R_{\geq0}$ a class $\KL$ function ($\beta \in \KL$) if for each fixed $t>0$, $\beta(\cdot,t)$ is of class $\K$ and for each fixed $s>0$, $\beta(s,\cdot)$ is decreasing and $\lim_{t\to\infty} \beta(s,t)=0$.
Given a locally Lipschitz continuous function \mbox{$\phi : \R^n \to \R^m$} and any compact set $B \subset \R^n$, we denote by $\mathcal{L}_\phi^B$ the local Lipschitz constant of $\phi$ on $B$. If $\phi$ is Lipschitz continuous, then we denote the Lipschitz constant as $\mathcal{L}_\phi$.


\subsection{Background}

We consider a control affine system of the form:
\begin{equation}
    \dot{x} = f(x) + g(x)u,
    \label{eq:sys}
\end{equation}
where $x \in \R^n$ is the state and $u \in \U$ is the control input, for which $\U \subseteq \R^m$ defines the set of admissible inputs. The functions $f : \R^n \to \R^n$ and $g : \R^n \to \R^{n \times m}$ are assumed to be locally Lipschitz continuous. Given a locally Lipschitz continuous controller $k : \R^n \to \R^m$, we define the closed-loop system with the input $u=k(x)$ to be:
\begin{equation}
    \dot{x} = f(x)+g(x)k(x) = f_{\text{cl}}(x). \label{eq:cl_sys}
\end{equation}
Given an initial condition $x_0 \in \R^n$ when $t=0$, we denote the solution to (\ref{eq:cl_sys}) at time $t=\tau$ by $x(\tau,x_0)$.

Safety is formally defined as restricting the solution of (\ref{eq:sys}) to a set $S \subset \R^n$ for all time $t \in \R_{\geq 0}$, i.e., it is synonymous with the forward invariance of the set $S$. We begin our exposition of safety by first defining set invariance and asymptotic stability.

\begin{definition}
    A set $S\subseteq \R^n$ is said to be forward invariant under the dynamics (\ref{eq:cl_sys}) if\footnote{We assume the solutions of (\ref{eq:cl_sys}) are forward complete.}:
    \begin{equation}
        x_0 \in S \implies x(t, x_0) \in S, \quad \forall t>0.
    \label{eq:fwd_invariance_def}
\end{equation}
\end{definition}

\begin{definition}
    A compact set $S \subset \R^n$ is said to be asymptotically stable with a region of attraction $\mathcal{D} \supseteq S$ for the dynamics (\ref{eq:cl_sys}) if:
    \begin{equation}
        \mathbf{d}(x(t, x_0), S) \leq \beta(\mathbf{d}(x_0, S), t), \quad \forall x_0 \in \mathcal{D}, \, \forall t>0,
        \label{eq:ass_stability_def}
    \end{equation}
    where $\beta \in \KL$ and $\mathbf{d}(x,S) = \inf_{y \in S}\Vert x - y \Vert$
    is the distance from a point $x$ to the set $S$.
\end{definition}

\begin{remark}
An asymptotically stable set is necessarily forward invariant. While safety is typically understood as forward invariance, in practice it is often desirable to obtain stronger guarantees such as asymptotic stability.
\end{remark}

\subsection{Control Barrier Functions}

To formulate safety in the context of CBFs, the safe set $S \subset \R^n$ is defined as the 0-superlevel set of a continuously differentiable function $h : \R^n \to \R$:
\begin{equation}
    S = \{ x \in \R^n : h(x) \geq 0 \}.
    \label{eq:safe_set_S}
\end{equation}
Here, $S$ is assumed to be non-empty and have no isolated points. 

\begin{definition}
    (\textit{Control Barrier Function (CBF) \cite{cbf_main}}) A continuously differentiable function, $h : \R^n \to \R$ is called a \textit{Control Barrier Function (CBF)} for the system (\ref{eq:sys}), if there exists $\alpha \in \Kie$ such that: \vspace{-5pt}
    \begin{equation}
        \sup_{u \in \U} L_fh(x)+L_gh(x)u+\alpha(h(x)) > 0,
        \label{eq:cbf_def}
    \end{equation} \vspace{-5pt}
    for all $x \in \R^n$.
\end{definition}
When $\U = \R^m$, the condition (\ref{eq:cbf_def}) can be equivalently stated as:
\begin{equation}
    L_gh(x)=0 \implies L_fh(x)+\alpha(h(x)) > 0.
    \label{eq:valid_cbf}
\end{equation}
It was shown in \cite{cbf_main,cbf_robustness} that, given a CBF $h$ and a controller $k : \R^n \to \mathcal U$ that satisfies:
    \begin{equation}
        L_fh(x)+L_gh(x) k(x) + \alpha(h(x))\geq 0, \quad \forall x \in \R^n, \label{eq:cbf_contr}
    \end{equation}
the safe set $S$ will be rendered forward invariant and asymptotically stable for the dynamics (\ref{eq:cl_sys}).

We also re-state Proposition 4 of \cite{cbf_robustness} as the following Lemma, which we shall refer to in the proof of our main Theorems.

\begin{lemma} \label{lem:h_invariance}
    Let $\tilde{h} : \R^n \to \R$ be a continuously differentiable function.
    For the system dynamics (\ref{eq:cl_sys}), if:
    \begin{equation}
        \dot{\tilde{h}}(x) + \alpha(\tilde{h}(x)) \geq 0, \quad \forall x\in \R^n,
        \label{eq:h_dot_condition}
    \end{equation}
    where $\alpha \in \Kie$, then the set $\tilde{S} = \{x \in \R^n: \tilde{h}(x)\geq 0 \}$ is forward invariant and asymptotically stable.
\end{lemma}
This lemma essentially states that given \textit{any} continuously differentiable function $\tilde{h}$ that satisfies (\ref{eq:h_dot_condition}), its 0-superlevel set will be forward invariant and asymptotically stable.



\section{Main Result}
\label{sct:main_result}

Generating a safe controller using a CBF requires knowledge of the true state $x$. In most practical scenarios, we typically only have an estimate $\hat{x}$ of the state $x$ and an estimation error bound $\delta  \in \R_{\geq0}$ such that: 
\begin{equation}
    \Vert x(t)-\hat{x}(t)\Vert\le \delta,\qquad \forall t\in \R_{\ge 0}.
    \label{eq:norm}
\end{equation}
Now the challenge lies in selecting a control input $u$ based solely on the information available in $\hat{x}$ to provide desirable safety guarantees.
The following problem statement formalizes a type of safety guarantee that may be desirable in practice:

\begin{problem}
    Given only an estimate $\hat{x}$ of the true state $x$, with the guarantee (\ref{eq:norm}), design a control law \mbox{$k:\R^n \to \mathcal U$} such that the closed-loop system defined by ($\ref{eq:sys}$) and $u=k(\hat{x})$ renders the set $S$ safe for small $\delta$ and renders an inflated set $S_\delta$\footnote{Throughout the paper we use the subscript notation \mbox{$S_\delta = \{ x \in \R^n : h(x) \geq -\delta \}$} to represent inflated versions of the safe set.} $\supseteq S$ safe for larger $\delta$.
\end{problem}

To address this problem, we first characterize a class of functions with desirable properties in the following definition.
\begin{definition} \label{def:rob_func}
    A continuous function $\rho : \R_{\geq 0} \to \R_{\geq 0}$ is called a robustness function if it satisfies: \vspace{-5pt}
    \begin{enumerate}
        \item $\rho(0) = 0$. \vspace{2pt}
        \item There exists $\varepsilon > 0$ such that:
        \begin{equation*}
            \inf_{y \in \R_{> 0}} \frac{\rho(y)}{y} = \varepsilon.
        \end{equation*}
        \item The function $\zeta : \R_{\geq 0} \to \R$, given by the expression:
        \begin{equation*}
        \inf_{y \in \R_{\geq 0}} ( \rho(y) - yd) = -\zeta(d),
        \end{equation*}
        is well defined, i.e., for every $d\in \R_{\ge 0}$ we have $\inf_{y \in \R_{\geq 0}} ( \rho(y) - yd)\in \R$.

    \end{enumerate}
\end{definition}

\begin{remark}
    Property 2 simply means that $\rho(y) \ge \varepsilon y$ for some $\varepsilon>0$, and property 3 simply requires that $\rho(y)-yd$ be lower bounded for any $d$.
\end{remark}

For example, the function $\rho(y) = \gamma_1 y + \gamma_2 y^2$ is a robustness function for $\gamma_1, \gamma_2 \in \R_{>0}$, since $\inf_{y \in \R_{> 0}} \frac{\rho(y)}{y} = \gamma_1$ and $\inf_{y \in \R_{\geq 0}} ( \rho(y) - yd) = \frac{-(\max \{0,  d-\gamma_1 \})^2}{4 \gamma_2}$. 
For a quicker reading of the paper, the reader may want to replace every usage of the class of functions in Definition \ref{def:rob_func} with $\rho(y)=\gamma_1 y+\gamma_2 y^2$.

\begin{remark}
    We note that property 3 above is closely related to the definition of the convex conjugate and Legendre–Fenchel transformation in convex analysis \cite{convex_book}. In fact $\zeta$ will be the convex conjugate of $\rho$ restricted to $\R_{\geq 0}$, and will be monotonically increasing in its argument.
\end{remark}

\begin{remark} \label{rem:delta_epsilon}
    Property 2 will also guarantee that $\zeta(d)=0$ for $d \leq \varepsilon$, since: \vspace{-5pt}
    \begin{align*}
        \inf_{y \in \R_{\geq 0}} ( \rho(y) - yd) &= \inf_{y \in \R_{> 0}} \left[ \left( \frac{\rho(y)}{y} - d\right) y \right] \\
        &\geq \inf_{y \in \R_{> 0}} [(\varepsilon - d)y] = 0, \vspace{-5pt}
    \end{align*}
    where the first equality follows from continuity at $y=0$.
    But $( \rho(y) - yd) |_{y=0} = 0$, and hence the lower bound provided by the preceding inequality is attained. This, combined with the fact that $\zeta$ is monotonically increasing will in fact mean that $\zeta(d) \geq 0$ for all $d \in \R_{\geq 0}$.
\end{remark}

Based on this characterization, we next define the notion of a Robust Control Barrier Function (R-CBF). 
\begin{definition}
    (\textit{Robust Control Barrier Function (R-CBF)}) A continuously differentiable function $h: \R^n \to \R$, is called a Robust Control Barrier Function (R-CBF) with robustness function $\rho : \R_{\geq 0} \to \R_{\geq 0}$, if there exists $\alpha \in \Kie$ such that: \vspace{-5pt}
    \begin{equation}
        \sup_{u \in \U} L_fh(x)+L_gh(x)u+\alpha(h(x)) > \rho(\Vert L_gh(x) \Vert),
        \label{eq:rcbf_def}
    \end{equation} \vspace{-5pt}
    for all $x \in \R^n$.
    \label{def:rcbf}
\end{definition}

We observe that for the case when $\U = \R^m$, condition (\ref{eq:rcbf_def}) reduces to:
\begin{equation*}
    L_gh(x)=0 \implies L_fh(x)+\alpha(h(x)) > 0,
\end{equation*}
due to property 1 of robustness functions in Definition \ref{def:rob_func}, thus proving that in the case of unconstrained input, any CBF will also be an R-CBF. As we shall see later, this will allow the enforcement of a stronger constraint, which provides more robust guarantees than the constraint imposed by regular CBFs without the need to construct a new function $h$.

A controller $k: \R^n \to \mathcal U$ is said to satisfy an R-CBF inequality if:
\begin{equation}
    L_fh(x)+L_gh(x)k(x)+\alpha(h(x)) \geq \rho(\Vert L_gh(x) \Vert),
    \label{eq:rcbf_contr}
\end{equation}
for all $x \in \R^n$, for some $\alpha\in \Kie$ and a robustness function $\rho$.

It is also worth noting that the existence of an R-CBF necessarily implies the existence of a continuous (and in fact, locally Lipschitz continuous) controller \mbox{$k : \R^n \to \mathcal U$} satisfying (\ref{eq:rcbf_contr}), as stated in the following Proposition.

\begin{proposition}
    Given any locally Lipschitz continuous controller \mbox{$k_d : \R^n \to \mathcal \R^m$}, and an R-CBF $h$, it is possible to construct a minimally invasive locally Lipschitz continuous controller \mbox{$k : \R^n \to \mathcal \R^m$}, satisfying (\ref{eq:rcbf_contr}) using the following Quadratic Program (QP):
    \begin{alignat}{2}
    k(x) =\; & \argmin _{u \in \R^m} \quad \Vert u - k_d(x) \Vert ^2 && \label{eq:Robust_QP}\\
     \quad  & \text{s.t.} \quad L_fh(x)+L_gh(x)u+\alpha(&&h(x)) \geq \nonumber\\&\qquad &&\rho(\Vert L_gh(x) \Vert).\nonumber
    \end{alignat}
\end{proposition}
The local Lipschitz continuity of the aforementioned controller follows from the strict inequality in Definition \ref{def:rcbf}, mirroring standard CBF proofs \cite{alyaseen2025continuityQP}, \cite{jankovic2018robust}.
Additionally, since \eqref{eq:Robust_QP} is a QP with a single linear constraint, a closed form expression for $k(x)$ can be obtained
similar to
\cite{cbf_robustness}.

When a state estimate is used with a control law in a closed-loop setting, an error in the state estimate will manifest itself as a deviation of the control input from the nominal value, i.e., an actuation error. Thus, in Section \ref{sct:act_err} we first show how robustness can be achieved against actuation errors and then in Section \ref{sct:meas_err}, we extend this to deal with errors in the state estimate.

\subsection{Actuation Errors} \label{sct:act_err}

We first prove the following intermediate result which extends the ideas of Input-to-State Safety (ISSf) for actuation uncertainty.
In this formulation, the true applied input will be $k(x)+d$, where $d \in \R^m$ is an actuator disturbance. Thus, the total input may no longer be restricted to $\mathcal U$.
\begin{theorem}
\label{thm:actuation_err}
Let $k : \R^n \to \mathcal U$ be a locally Lipschitz continuous controller that satisfies (\ref{eq:rcbf_contr}) 
and let the closed-loop system be defined by (\ref{eq:sys}) with $u(t)=k(x(t))+d(t)$, where the signal $d:\R_{\ge 0}\to \R^m$ satisfies $\Vert d(t)\Vert\le \overline{d}$ for all $t\in \R_{\ge 0}$ and some $\overline{d}\in \R_{\ge 0}$. Then, there exists an $\varepsilon \in \R_{>0}$ such that:
\begin{enumerate}
    \item if $\overline{d} \leq \varepsilon$, the set $S = \{ x \in \R^n : h(x) \geq 0 \}$ is asymptotically stable and thus also forward invariant.
    \item if $\overline{d} > \varepsilon$
    there exists $\xi \in \R_{>0}$ such that the inflated safe set $S_\xi = \{ x \in \R^n : h(x) \geq -\xi \}$ is asymptotically stable and thus also forward invariant.
\end{enumerate}
\end{theorem}
\begin{proof}
    Let us use the notation $\dot{h}(x,u)$ to denote $\dot{h}(x,u) = L_fh(x) + L_gh(x)u$. Now, we evaluate $\dot{h}(x,u)$ with $u=k(x)+d$:
    \begin{align}
        \dot{h}(x,u) &= L_fh(x)+L_gh(x)(k(x)+d) \nonumber\\
        & \geq -\alpha(h(x)) + \rho(\Vert L_gh(x)\Vert) - \Vert L_gh(x) \Vert \overline{d}. \label{eq:h_dot_proof_1}
    \end{align}
    When $L_gh(x)=0$, we have $\dot{h}(x,u) + \alpha(h(x))>0$. For the other case, we rearrange the terms to obtain:
    \begin{equation*}
        \dot{h}(x,u) + \alpha(h(x)) \geq \Vert L_gh(x) \Vert \left( \frac{\rho(\Vert L_gh(x)\Vert)}{\Vert L_gh(x) \Vert} - \overline{d}\right),
    \end{equation*}
    Now, by property 2 of Definition \ref{def:rob_func}, there exists $\varepsilon \in \R_{>0}$ such that:
    \begin{equation*}
        \dot{h}(x,u) + \alpha(h(x)) \geq \Vert L_gh(x) \Vert (\varepsilon - \overline{d}).
    \end{equation*}
    When $\overline{d} \leq \varepsilon$, we have $\dot{h}(x,u) + \alpha(h(x)) \geq0$ for all $x \in \R^n$, and by Lemma \ref{lem:h_invariance}, the set $S$ will be asymptotically stable and thus also forward invariant.
    
    To obtain the second case, we note that (\ref{eq:h_dot_proof_1}) implies:
    \begin{align*}
        \dot{h}(x,u) + \alpha(h(x)) &\geq  \inf_{x \in \R^n} ( \rho(\Vert L_gh(x)\Vert) - \Vert L_gh(x)\Vert \overline{d})\\
        &\geq -\zeta(\overline{d}),
    \end{align*}
    where the last inequality follows from property 3 of Definition \ref{def:rob_func}. Now defining $h_\zeta(x) = h(x)-\alpha^{-1}(-\zeta(\overline{d}))$, we obtain:
    \begin{equation*}
        \dot{h}_\zeta(x,u)+\tilde{\alpha}(h_\zeta(x)) \geq 0,
    \end{equation*}
    where $\tilde{\alpha}(h_\zeta(x)) =\alpha(h_\zeta(x) + \alpha^{-1}(-\zeta(\overline{d})))+\zeta(\overline{d})$, and $\tilde{\alpha} \in \Kie$ \cite{compendium_class_k}. It follows from Lemma \ref{lem:h_invariance} that the set $S_\xi = \{ x \in \R^n : h(x) \geq \alpha^{-1}(-\zeta(\overline{d}))\}$ is asymptotically stable and thus also forward invariant.
\end{proof}

We observe here that property 2 of Definition \ref{def:rob_func} yields the first result of Theorem \ref{thm:actuation_err}: invariance of the original safe set under small actuation errors, while property 3, on the other hand, yields the second result: invariance of an inflated set under larger actuation errors.
Since the invariant set is given by \mbox{$S_\xi = \{ x \in \R^n : h(x) \geq \alpha^{-1}(-\zeta(\overline{d}))\}$}, the shape of the function $\zeta$, determines the nature of the safety degradation.

\begin{remark}
    For \mbox{$\rho(y) = \gamma_1y+\gamma_2y^2$}, 
    by appropriately substituting for $\varepsilon$ and $\zeta (\overline d)$ in the proof of Theorem \ref{thm:actuation_err}, we conclude that the invariant set is given by:
    \begin{equation}
        \left\{ 
        x \in \R^n : h(x) \geq \alpha^{-1} \left( -\frac{(\max(0, \overline d-\gamma_1))^2}{4 \gamma_2} \right)
        \right\}.
        \label{eq:set_expr_quad}
    \end{equation}
    Clearly, when $\overline{d} \leq \gamma_1$, the original safe set $S$ is forward invariant.
\end{remark}

\begin{remark}
    A popular modification for robustness against actuation uncertainty is to use an ISSf-CBF \cite{issf}. This also yields a strengthened CBF inequality, where the modification term $\rho(y)=\gamma y^2$ with $\gamma \in \R_{>0}$, satisfies only properties 1 and 3 of Definition \ref{def:rob_func}. Thus, in general, an ISSf-CBF will only guarantee invariance of an inflated set and not of the original safe set $S$.
    Alternatively, \cite{jankovic2018robust} strengthens the inequality using $\rho(y)=\gamma y$ with $\gamma \in \R_{>0}$. Because this choice satisfies only properties 1 and 2, it guarantees invariance of $S$ only for $\overline{d} \leq \gamma$ and lacks a graceful degradation of safety via inflated invariant sets.
    The R-CBF framework unifies the robust guarantees of both \cite{issf} and \cite{jankovic2018robust}, providing a mathematically broader result by axiomatically characterizing the class of functions $\rho$ required to achieve these properties.
    \label{rem:issf-cbf}
\end{remark}

For the case of unconstrained input, if one possesses a CBF for a system, one may employ a QP as in (\ref{eq:Robust_QP}) to construct a controller that satisfies the stronger R-CBF condition. Since any CBF will also be an R-CBF, the feasibility of this QP is always guaranteed. It is important to note, however, that the controllers generated by the CBF and R-CBF inequalities are different, since the R-CBF inequality imposes a stronger constraint.

The specific choice of the robustness function determines both the uncertainty threshold up to which the original safe set remains forward invariant, and the rate of set inflation for disturbances exceeding this limit. 
For instance, when employing the function \mbox{$\rho(y) = \gamma_1 y + \gamma_2 y^2$}, it is evident from \eqref{eq:set_expr_quad} that $\gamma_1$ governs the threshold for strict invariance, while $\gamma_2$ dictates the subsequent rate of inflation. 
In the absence of input constraints, a designer may choose arbitrarily large values for $\gamma_1$ and $\gamma_2$ while still obtaining a valid R-CBF. 
In contrast, for input constrained systems, a practical tuning heuristic is to initialize these parameters with small values, gradually incrementing them until either the desired robustness margin is achieved or the input constraints are reached. 
Although selecting larger parameter values provides stronger theoretical robustness guarantees, it inevitably induces more conservative closed-loop behavior. 
This conservatism becomes particularly pronounced when the actual actuation disturbances encountered during operation are significantly milder than the anticipated worst-case bounds.

\subsection{Measurement Errors} \label{sct:meas_err}

In this section we show how the analysis of the previous section can be extended to provide guarantees against measurement uncertainty. At the technical level, we leverage the compactness of the super-level sets of $h$ and the continuity of the controller to convert measurement uncertainty into actuation errors.

\begin{theorem}
Let $k : \R^n \to \mathcal{U}$ be a continuous controller that satisfies (\ref{eq:rcbf_contr})  
and let the closed-loop system be defined by (\ref{eq:sys}) with $u=k(\hat{x})$, where $\hat{x}$ satisfies (\ref{eq:norm}).
Then, for any $S_\beta \supseteq S$ that is a compact super-level set of $h$, 
there exists $\epsilon_1, \epsilon_2 \in \R_{>0}$, with $\epsilon_1<\epsilon_2$ such that:
\begin{enumerate}
    \item if $\delta \leq \epsilon_1$, the set $S = \{ x \in \R^n : h(x) \geq 0 \}$ is asymptotically stable with a region of attraction containing $S_\beta$, and thus also forward invariant.
    \item if $\epsilon_1 < \delta \leq \epsilon_2$, there exists $\xi \in \R_{>0}$ such that the inflated safe set $S_\xi = \{ x \in \R^n : h(x) \geq -\xi \} \subseteq S_\beta$ is asymptotically stable with a region of attraction containing $S_\beta$, and thus also forward invariant.
\end{enumerate}
\label{thm:main}
\end{theorem}

\begin{proof}
We first show that given any compact $S_\beta$, there exists $\epsilon_2 \in \R_{>0}$ such that $S_\beta$ is forward invariant when $\delta \leq \epsilon_2$.

Consider $\dot{h}(x, u)$ with $u = k(\hat{x})$:
\begin{align}
    \dot{h}(x, k(\hat{x})) &= L_fh(x)+L_gh(x)k(\hat{x}) \nonumber \\
        &= L_fh(x)+L_gh(x)k(x)\nonumber
        \\& \qquad
        +L_gh(x)(k(\hat{x}) - k(x)) \nonumber\\
        &\geq \rho(\Vert L_gh(x) \Vert ) - \alpha(h(x)) \nonumber
        \\&\qquad - \Vert L_gh(x) \Vert \Vert k(\hat{x}) - k(x) \Vert, \label{eq:h_dot_bound}
\end{align}
where the first inequality follows from (\ref{eq:rcbf_contr}) and properties of inner products.

Next, we define the function $\tilde{\sigma}_\beta : \R_{\geq0} \to \R_{\geq0}$ as:
\begin{equation}
    \tilde{\sigma}_\beta (r) = \sup_{\Vert e \Vert \leq r} \left( \sup_{z \in S_\beta} \Vert k(z+e)-k(z) \Vert \right).
    \label{eq:sigma_beta_sup}
\end{equation}
We note that since $k$ is a continuous function, the supremum over $(x,e)$ in the compact set $S_\beta \times B_r(0)$ is well defined. Since $\tilde{\sigma}_\beta$ is a non-decreasing function with the property $\tilde{\sigma}_\beta(0)=0$, it can be upper bounded by a class $\K$ function $\sigma_\beta : \R_{\geq0} \to \R_{\geq0}$ \cite{compendium_class_k}.

Now, noting that $\hat{x} = x+e$ where $\Vert e \Vert \leq \delta$, we observe that:
\begin{equation}
    \Vert k(\hat{x})-k(x) \Vert \leq \tilde{\sigma}_\beta(\delta) \leq \sigma_\beta(\delta).
    \label{eq:sigma_beta}
\end{equation}
We substitute this in (\ref{eq:h_dot_bound}) to obtain:
\begin{align}
    \dot{h}(x, k(\hat{x})) &\geq \rho(\Vert L_gh(x) \Vert ) - \sigma_\beta(\delta) \Vert L_gh(x) \Vert - \alpha(h(x)) \nonumber\\
    &\geq -\zeta(\sigma_\beta(\delta))-\alpha(h(x)), \label{eq:h_dot_alpha}
\end{align}
where the second inequality follows from property 3 of Definition \ref{def:rob_func}.
Letting $\tilde h(x) = h(x) + \beta$ and \mbox{$\tilde \alpha (\tilde h(x)) = \alpha(\tilde h(x)-\beta)-\alpha(-\beta)$}, where $\tilde{\alpha} \in \Kie$ \cite{compendium_class_k}, \eqref{eq:h_dot_alpha} now implies:
\begin{align}
    \dot{\t h}(x) + \t \alpha(\t h(x)) &\geq -\zeta(\sigma_\beta(\delta)) - \alpha(-\beta),
\end{align}
meaning:
\begin{equation}
    \zeta(\sigma_\beta(\delta)) \leq -\alpha(-\beta) \implies \dot{\t h}(x) + \t \alpha(\t h(x)) \geq 0.
\end{equation}
Since $\zeta$ is monotonically increasing, and $\sigma_\beta \in \K$, the antecedent of the implication is equivalent to $\delta \leq \epsilon_2$, where:
\begin{equation*}
    \epsilon_2 = \sup \{ \delta \in \R_{\geq 0} : \zeta(\sigma_\beta(\delta)) \leq -\alpha(-\beta) \}.
\end{equation*}
Hence, by Lemma \ref{lem:h_invariance}, $S_\beta$ is forward invariant for $\delta \leq \epsilon_2$.

Finally, within this invariant set, we have: \vspace{-5pt}
\begin{equation*}
    k(\hat{x}) = k(x) + d(t), \vspace{-5pt}
\end{equation*} 
where $\Vert d(t) \Vert \leq \sigma_\beta(\delta)$ for all $t \in \R_{>0}$. 
Thus, we may invoke Theorem \ref{thm:actuation_err} within $S_\beta$, with $\overline{d} = \sigma_\beta(\delta)$ to conclude that when $\delta \leq \epsilon_1 =  \sigma_\beta^{-1} \left(\inf_{y\in\R_{>0}}\frac{\rho(y)}{y}\right)$, the set $S$ is forward invariant and asymptotically stable, 
and when \mbox{$\epsilon_1 < \delta < \epsilon_2$}, the set: 
\begin{equation*}
    S_\xi = \{ x \in \R^n : h(x) \geq \alpha^{-1}(-\zeta(\sigma_\beta(\delta)))\} \subseteq S_\beta,
\end{equation*}
is forward invariant and asymptotically stable.
\end{proof}

As stated in the preceding theorem, a controller satisfying an R-CBF inequality (\ref{eq:rcbf_contr}), will guarantee that for small amounts of state uncertainty, the invariance and asymptotic stability of the original safe set is retained.
When the uncertainty grows beyond this threshold, we retain guarantees for an inflated safe set in an ISSf-like manner.
Essentially, this means that, given a controller satisfying an R-CBF constraint, we can safely synthesize the control input using the estimate $\hat{x}$ rather than the true state $x$. 
In practice, this can be implemented by solving the QP \eqref{eq:Robust_QP} with $\hat x$ in place of $x$.

Similarly to the case of actuation disturbances, the specific choice of robustness function will determine the exact values of $\epsilon_1$ and $\epsilon_2$ in Theorem \ref{thm:main}.
For \mbox{$\rho(y)=\gamma_1y+\gamma_2y^2$}, a general rule of thumb is that larger parameter values yield stronger theoretical guarantees against state uncertainty. 
However, unlike in the case of actuation disturbances, this relationship is not strictly monotonic. Employing excessively large values for $\gamma_1$ and $\gamma_2$ induces high sensitivity in the control law, causing $k$ to fluctuate drastically in response to minute state variations. Mathematically, this implies that the sensitivity bound $\tilde{\sigma}_\beta(r)$ in \eqref{eq:sigma_beta_sup} evaluates to a disproportionately large magnitude even for small $r$, which in turn causes the thresholds $\epsilon_1$ and $\epsilon_2$ to shrink.
Consequently, an effective tuning heuristic is to initialize $\gamma_1$ and $\gamma_2$ with small values, iteratively increasing them until a satisfactory balance between robustness and closed-loop sensitivity is achieved.

\subsection{Extensions}

We now extend the results of Theorem \ref{thm:main} for the case of an unbounded safe set\footnote{In this case, the definition of asymptotic stability needs to be strengthened by requiring forward completeness of solutions.} given the existence of a Lipschitz continuous controller satisfying (\ref{eq:rcbf_contr}), as stated in the following Corollary.

\begin{corollary}
Let $k : \R^n \to \mathcal U$ be a Lipschitz continuous controller, with Lipschitz constant $\mathcal{L}_k$, that satisfies (\ref{eq:rcbf_contr}).
Then, for the closed-loop system defined by (\ref{eq:sys}) with $u=k(\hat{x})$, where $\hat{x}$ satisfies (\ref{eq:norm}):
\begin{enumerate}
    \item when $\delta \leq \frac{\varepsilon}{\mathcal{L}_k}$, the set $S = \{ x \in \R^n : h(x) \geq 0 \}$ is forward invariant and globally asymptotically stable.
    \item when $\delta > \frac{\varepsilon}{\mathcal{L}_k}$, the inflated set defined by \linebreak {$S_\xi = \left\{ x \in \R^n : h(x) \geq \alpha^{-1}\left( -\zeta ( \mathcal{L}_k \delta ) \right) \right\}$} is forward invariant and globally asymptotically stable.
\end{enumerate}  
\label{crl:glob}
\end{corollary}
Note however, that the existence of a Lipschitz continuous controller is not necessarily guaranteed by the existence of an R-CBF, and will depend on the nature of the system and choice of R-CBF.

We also note that robustness functions may be modified in a ``tunable'' manner similar to \cite{tissf}, to reduce conservatism when the trajectories of the system lie further inside the interior of the safe set.
The following corollary demonstrates this for the specific choice of robustness function, $\rho(y) = \gamma_1y + \gamma_2y^2$, where $\gamma_1, \gamma_2 \in \R_{>0}$.

\begin{corollary}
    Let $k:\R^n \to \mathcal U$ be a continuous controller that satisfies:
    \begin{align}
        L_f&h(x)+L_gh(x)k(x)+\alpha(h(x)) \\
        &\geq \frac{\gamma_1}{\varepsilon_1(h(x))} \Vert L_gh(x) \Vert + \frac{\gamma_2}{\varepsilon_2(h(x))} \Vert L_gh(x) \Vert^2, \nonumber
    \end{align}
    for all $x \in \R^n$, where $\varepsilon_1, \varepsilon_2 : \R \to \R$ are continuous functions that satisfy $\varepsilon_i(r)=1$ when $r\leq 0$, and \mbox{$\frac{d \varepsilon_i}{dr} > 0$} when $r>0$, for $i=1,2$. Then the conclusions of Theorem \ref{thm:main} will hold for the closed-loop system defined by (\ref{eq:sys}) with $u=k(\hat{x})$.
\end{corollary}

We conclude this section by observing that an R-CBF will provide robust guarantees against both actuation errors and state uncertainty simultaneously, as summarized in the following Corollary.

\begin{corollary}
    Let $k : \R^n \to \mathcal U$ be a continuous controller that satisfies (\ref{eq:rcbf_contr})  
    and let the closed-loop system be defined by (\ref{eq:sys}) with $u(t)=k(\hat{x}(t))+d(t)$, 
    where $\hat{x}(t)$ satisfies (\ref{eq:norm}) and the signal $d:\R_{\ge 0}\to \R^m$ satisfies $\Vert d(t)\Vert\le \overline{d}$ for all $t\in \R_{\ge 0}$ and some $\overline{d}\in \R_{\ge 0}$.
    Then, for any $S_\beta \supseteq S$ that is a compact super-level set of $h$, 
    there exists $\epsilon_1, \epsilon_2 \in \R_{>0}$, with $\epsilon_1<\epsilon_2$, and $\sigma_\beta \in \mathcal K$ such that:
    \begin{enumerate}
        \item if $\sigma_\beta(\delta) + \overline{d} \leq \epsilon_1$, the set $S = \{ x \in \R^n : h(x) \geq 0 \}$ is asymptotically stable with a region of attraction containing $S_\beta$, and thus also forward invariant.
        \item if $\epsilon_1 < \sigma_\beta(\delta) + \overline{d} \leq \epsilon_2$, there exists $\xi \in \R_{>0}$ such that the inflated set $S_\xi = \{ x \in \R^n : h(x) \geq -\xi \} \subseteq S_\beta$ is asymptotically stable with a region of attraction containing $S_\beta$, and thus also forward invariant.
    \end{enumerate}
\end{corollary}

\begin{remark}
    Similar to the case of regular CBFs, the controller \eqref{eq:Robust_QP} may not necessarily preserve the stability properties of the desired controller $k_d(x)$. 
    If the equilibrium point of the closed-loop system with $u=k_d(x)$ lies strictly inside the safe set, for an appropriate choice of class $\K$ function $\alpha$, the R-CBF constraint remains inactive in a local neighborhood, thereby ensuring the stability of the closed-loop system under the controller \eqref{eq:Robust_QP} within this neighbourhood. 
    However, establishing global stability guarantees remains a highly non-trivial, open problem even within the regular CBF literature, and is left as a direction for future research.
\end{remark}

\subsection{Simulations} \label{sct:simulations}

\begin{figure}
    \centering
    \includegraphics[width=1\linewidth]{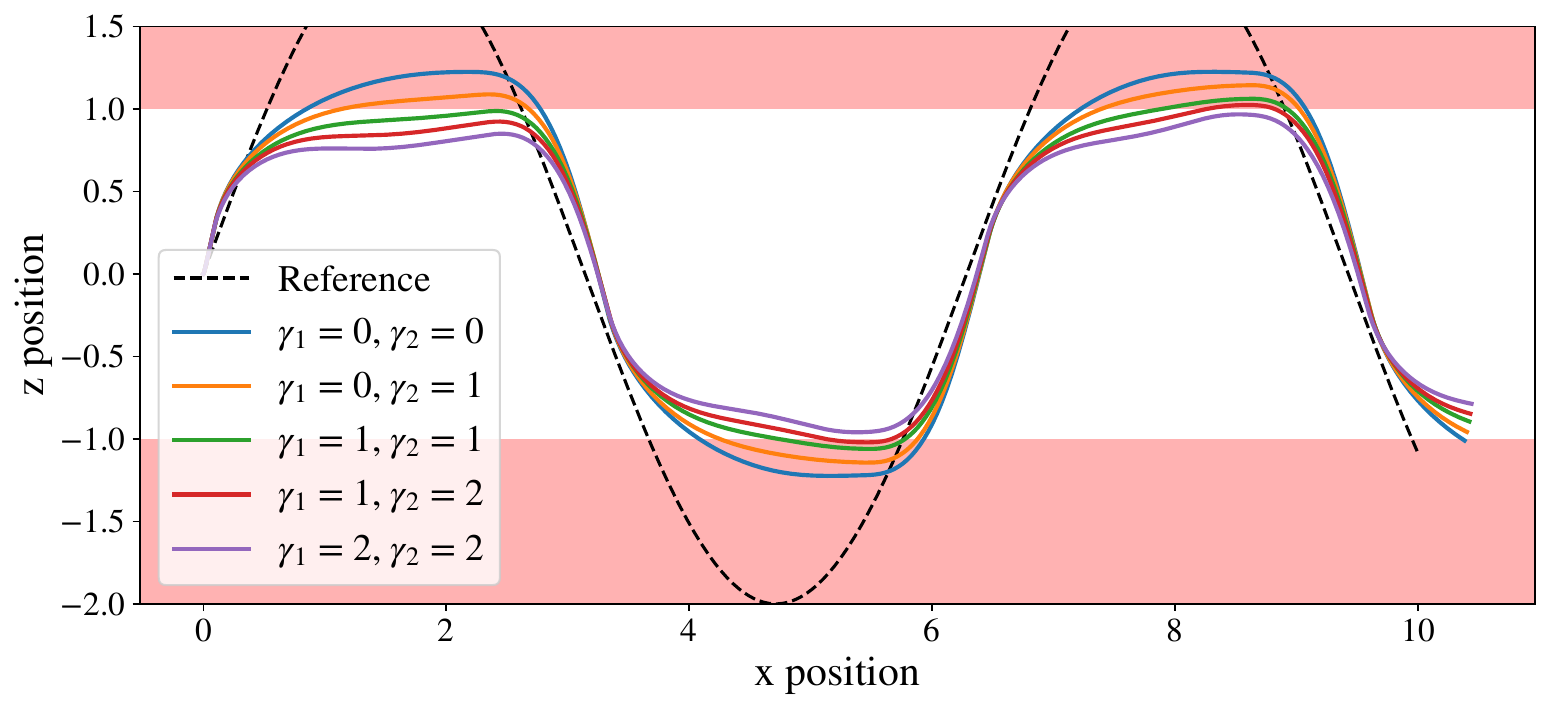}
    \caption{Closed-loop trajectories of the unicycle model for various R-CBFs under state uncertainty $\delta=0.25$.}
    \label{fig:unicycle_traj}
\end{figure}

\begin{figure}
    \centering
    \includegraphics[width=1\linewidth]{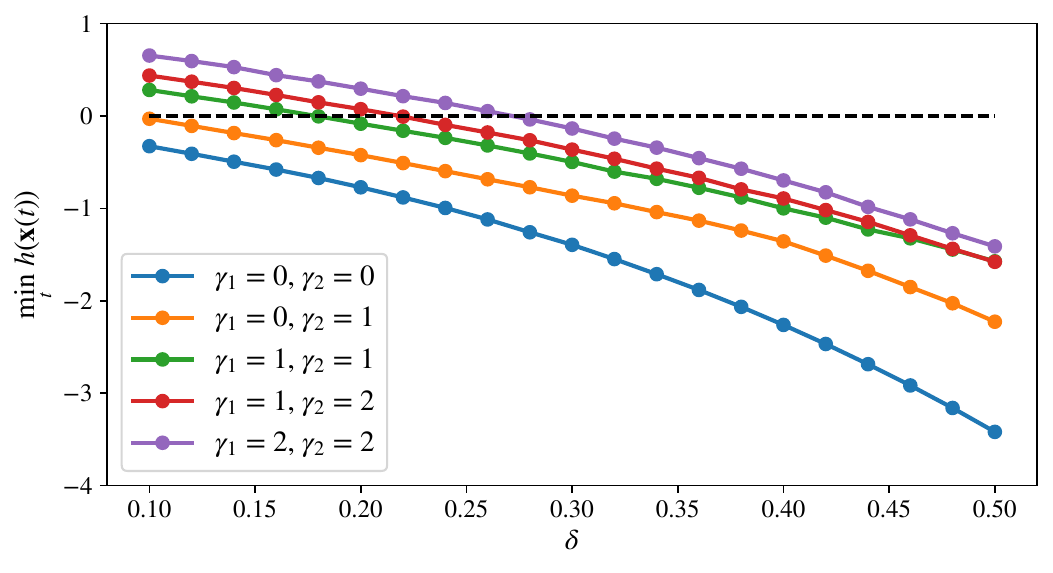}
    \caption{Minimum value of $h(x(t))$ over the entire time horizon $[0, 10]$ for varying levels of $\delta$.}
    \label{fig:unicycle_vs_d}
\end{figure}

\noindent \textbf{Example 1 :} We first illustrate the application of the R-CBF framework to the unicycle model, to show how the design of the robustness function affects the closed loop guarantees.
Consider the system dynamics:
\begin{equation}
    \begin{aligned}
        \dot x &= v \cos \theta \\
        \dot z &= v \sin \theta \\
        \dot \theta &= \omega \\
        \dot v &= a,
    \end{aligned}
\end{equation}
where $\textbf{x} = (x, z, \theta, v) \in \R^4$ is the state and $\textbf{u} = (\omega, a) \in \R^2$ is the input.
We seek to restrict the motion of the robot to a corridor region $\{ (x, z)\in \R^2 : |z| \leq 1 \}$. 
Since the safety constraint is of higher relative degree, we enforce it using the higher order CBF \cite{nguyen2016exponential} \mbox{$h(\mathbf{x})= 1-z^2 -zv \sin \theta$}.
We employ the robustness function $\rho(y)=\gamma_1 y + \gamma_2 y^2$, where $\gamma_1, \gamma_2 \in \R_{>0}$, noting that the R-CBF framework is applicable to any CBF, including higher-order variants.
To stress test our method, we artificially corrupt the state estimate with an adversarial error
$\mathbf{\hat {x}} = \textbf{x} - \mathbf{d}$, where $\textbf{d} \in \R^4$ with components defined as $\textbf{d}_i = \delta \tanh(\textbf{x}_i)$,
for varying levels of $\delta \in \R_{>0}$.
This specific disturbance ensures that $\mathbf{\hat{x}}$ lies further in the interior of the safe set than the actual state \textbf{x}, effectively tricking the controller with an artificially optimistic safety margin.
The nominal controller $k_d(\textbf{x})$ is designed to track a reference trajectory in $\R^2$ given by $x_\text{ref}(t)=t$ and $z_\text{ref}(t)=2 \sin(t)$, parts of which lie outside the safe set.
Finally, we generate the closed-loop control input by evaluating \eqref{eq:Robust_QP} at the estimated state.

Fig. \ref{fig:unicycle_traj} shows the evolution of closed loop trajectories for $\delta=0.25$.
The reference trajectory tracked by the nominal controller is denoted by a dashed black line, and the unsafe region is shaded in pink.
The blue curve corresponding to $\gamma_1=\gamma_2=0$ essentially means that a regular CBF is being employed with the state estimate, and therefore clearly violates the safety requirement.
The orange curve with $\gamma_1=0$ and $\gamma_2=1$ results in an ISSf-CBF \cite{issf}, and shows less safety violation but still fails to contain the trajectories to the safe set.
The figure clearly illustrates that using higher values for $\gamma_1$ and $\gamma_2$, where both terms are non-zero, results in robust safety.

To further evaluate this method, Fig. \ref{fig:unicycle_vs_d} explores the fundamental trade-off between safety guarantees and controller conservatism by plotting the minimum value of $h(\mathbf{x}(t))$ over the simulation horizon $[0, 10]$ as a function of $\delta$. 
As expected, the baseline CBF (blue curve) violates safety for disturbances as small as $0.1$. While the ISSf-CBF (orange curve) provides a more graceful degradation of safety (as demonstrated in \cite{issf}), it fails to render the original safe set invariant for any $\delta \geq 0.1$, as discussed in Remark \ref{rem:issf-cbf}.
In contrast, the R-CBFs ($\gamma_1, \gamma_2 > 0$) successfully maintain the invariance of the original safe set up to a critical uncertainty threshold. 
However, it is evident that selecting excessively high values for $\gamma_1$ and $\gamma_2$ induces highly conservative behavior when the actual uncertainty is significantly lower than the maximum allowable threshold. 
For instance, when $\gamma_1=\gamma_2=2$ and $\delta=0.1$, Fig. \ref{fig:unicycle_vs_d} shows that $h(x(t))$ is always close to $1$. This essentially means that the robot stays close to the center of the corridor and does not attempt to track the reference for fear of safety violation.

\begin{figure}[t]
    \centering
    \begin{subfigure}[t]{0.49\linewidth}
        \centering
        \includegraphics[width=\linewidth]{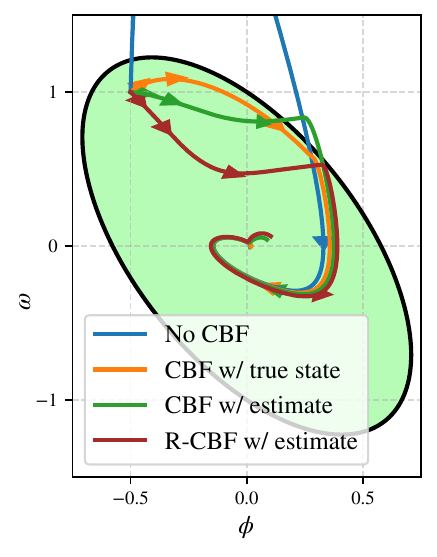}
    \end{subfigure}
    \hfill
    \begin{subfigure}[t]{0.49\linewidth}
        \centering
        \includegraphics[width=\linewidth]{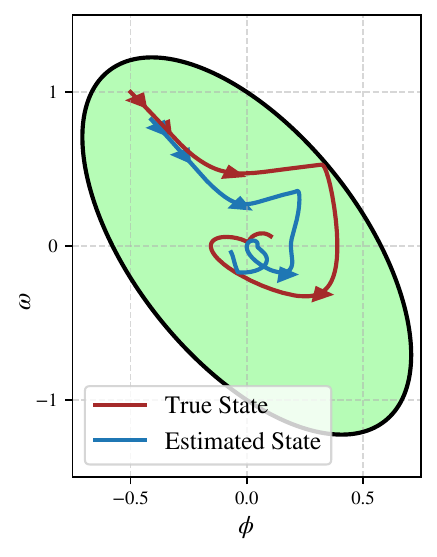}
    \end{subfigure}
    \hfill
    \caption{Left: Supervision by a regular CBFs vs R-CBF (\ref{eq:example_cbf}). Right: True state and state estimate under R-CBF. 
    }
    \label{fig:segway_lqr}
\end{figure}

\begin{figure}[t]
    \centering
    \begin{subfigure}[t]{0.49\linewidth}
        \centering
        \includegraphics[width=\linewidth]{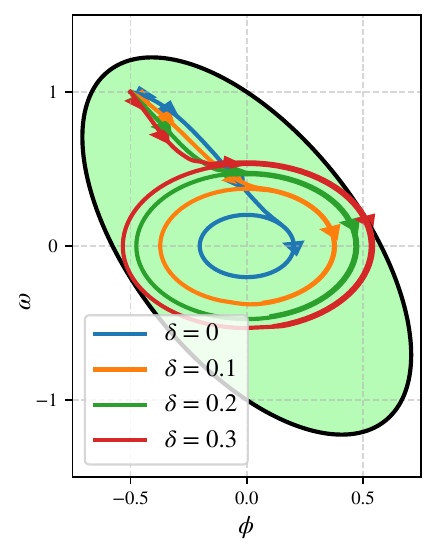}
    \end{subfigure}
    \hfill
    \begin{subfigure}[t]{0.49\linewidth}
        \centering
        \includegraphics[width=\linewidth]{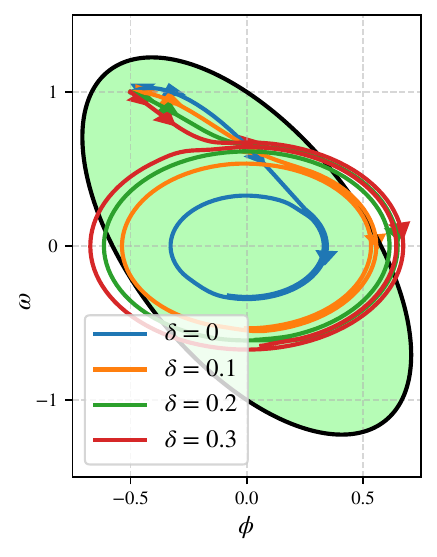}
    \end{subfigure}
    \hfill
    \caption{Varying levels of state uncertainty when (\ref{eq:example_cbf}) is used as an R-CBF (left) and an ISSf-CBF (right).}
    \label{fig:segway_issf}
\end{figure}

\noindent \textbf{Example 2 :} 
Consider the dynamic model of a Segway presented in \cite{dynamic_env_segway} with the state $\textbf{x}=(p, \phi, v, \omega)^T$ and the CBF:
\begin{equation}
    h(\textbf{x}) = 1-(3 \phi^2+2\phi \omega + \omega^2).
    \label{eq:example_cbf}
\end{equation}
The safe set defined by the CBF corresponds to maintaining small angular displacement and angular velocity as shown by the green shaded area in Fig. \ref{fig:segway_lqr} and \ref{fig:segway_issf}.
We simulate the closed-loop dynamics of the Segway with initial condition $\textbf{x}_0=(-4, -0.5, 0, 1)$ under an unsafe LQR controller and analyze the impact of state uncertainty.
Fig. \ref{fig:segway_lqr} (left) illustrates four cases: 1) no CBF is applied; 2) CBF (\ref{eq:example_cbf}) is used with the true state; 3) CBF (\ref{eq:example_cbf}) is used with a state estimate; 4) an R-CBF (\ref{eq:example_cbf}) with robustness function $\rho(y)=y+y^2$
is used.
Fig. \ref{fig:segway_lqr} (right) shows the true state of the system alongside the estimated state, for the case with the R-CBF.
To rigorously evaluate the controller, we artificially corrupt the state estimate as $\mathbf{\hat x} = \textbf{x} - (\delta, \delta\cos \theta, \delta, \delta \sin \theta)$, where $\theta=\text{arctan2}(\omega, \phi)$.
Once again, this specific perturbation ensures that
the estimated state lies further inside the interior of the safe set than the true state giving the controller a false sense of safety.
It is evident from Fig. \ref{fig:segway_lqr} that under this corrupted estimate (for $\delta=0.2$), the R-CBF guarantees safety of the original set whereas the regular CBF shows safety violations.

We also compare the R-CBF with a regular ISSf-CBF implementation 
(where $\rho(y)=y^2$)
in Fig. \ref{fig:segway_issf}. 
It can be seen that while the ISSf framework guarantees invariance of an inflated safe set, it does not ensure the invariance of the original safe set, even for small values of uncertainty, as discussed in Remark \ref{rem:issf-cbf}.

\section{Comparison to MR-CBFs} \label{sct:comparison_to_mrcbfs}

In \cite{mrcbf2021guaranteeing}, the authors adopt a different approach to 
propose a sufficient condition to guarantee safety in the presence of state uncertainty
by introducing the notion of Measurement Robust Control Barrier Functions (MR-CBFs). The authors assume that $L_fh$, $ L_gh$ and $\alpha \circ h$ are Lipschitz continuous on $S$, and assume prior knowledge of a function $\delta : \R^n \to \R_{>0}$ such that,
$\Vert \hat{x} - x \Vert \leq \delta(\hat{x})$.
The set of possible state estimates when the true state is in $S$ is then defined as
    \mbox{$\hat{S}=\{ \hat{x} \in \R^n : \hat{x}+y \in S, \, \forall y \in B_{\delta(\hat{x})}(0) \}$.}
A controller $k : \R^n \to \R^m$ is said to satisfy an MR-CBF inequality if:
\begin{align}
L_fh(\hat{x}) + &L_gh(\hat{x})k(\hat{x})+\alpha(h(\hat{x})) \label{eq:mr_cbf_controller}\\
& -\delta(\hat{x})\big(\Lip_{L_fh}+\Lip_{\alpha \circ h} + \Lip_{L_gh} \Vert k(\hat{x}) \Vert_2 \big) \geq 0, 
\nonumber
\end{align}
for all $\hat{x} \in \hat{S}$.
Under unconstrained inputs, the authors show that $h:\R^n \to \R$ is an MR-CBF iff:
\begin{equation}
\begin{array}{r}
     \Vert L_gh(\hat{x}) \Vert_2 \leq \quad\\
      \delta(\hat{x}) \Lip_{L_gh}
\end{array}
      \implies 
\begin{array}{r}
      L_fh(\hat{x})+\alpha(h(\hat{x})) \geq \quad \\
      \delta(\hat{x})(\Lip_{L_fh}+\Lip_{\alpha \circ h}).
\end{array}
\label{eq:mrcbf_nec}
\end{equation}

\begin{remark}
    The condition imposed on a function $h$ to be an MR-CBF is stricter than that imposed on $h$ to be a CBF. In other words, there are systems for which a CBF exists, but no MR-CBF exists.
\end{remark}
It is also shown that given a CBF $h$, one may extend it to form an MR-CBF if the following holds for all $\hat{x} \in \hat{S}$:
\begin{equation}
    \delta(\hat{x}) \leq \max \left\{\frac{\Vert L_gh(\hat{x}) \Vert_2}{\mathcal{L}_{L_gh}}, \frac{L_fh(\hat{x})+\alpha(h(\hat{x}))}{\mathcal{L}_{L_fh}+\mathcal{L}_{\alpha \circ h}} \right\}.
    \label{eq:mrcbf_delta_limit}
\end{equation}

We now compare the R-CBF to the existing MR-CBF approach. 
We begin by noting that the MR-CBF framework requires prior knowledge of the function $\delta$ for checking if a function is an MR-CBF and for computing its induced controller, whereas the proposed R-CBF does not require the knowledge of a bound on the uncertainty.
In the MR-CBF case, if the estimation error is too large, the condition (\ref{eq:mrcbf_delta_limit}) will be violated and no MR-CBF will exist. 
In such cases, one might attempt to design an MR-CBF for a state uncertainty lower than the true value, but the MR-CBF framework does not inherently provide robustness guarantees for this approach.

Next, we compare the applicability of the two approaches.
Consider a CBF $h$ such that $L_fh$, $ L_gh$ and $\alpha \circ h$ are locally Lipschitz continuous and their restriction to $S$ is Lipschitz continuous. Assume that there exists some $\delta(\hat{x})=d>0$ such that (\ref{eq:mrcbf_delta_limit}) is satisfied, then $h$ will be an MR-CBF. 
Since (\ref{eq:mrcbf_nec}) is a stronger condition than (\ref{eq:valid_cbf}), $h$ being an MR-CBF will imply that it is also an R-CBF, though the converse is not always true.
Thus, where an MR-CBF controller exists, we also have a continuous R-CBF controller.

However, this does not imply that any $k$ that satisfies (\ref{eq:mr_cbf_controller}) will also satisfy (\ref{eq:rcbf_contr}). 
This is because the proposed R-CBF formulation imposes additional robustness requirements beyond those of the MR-CBF framework. In particular, it is defined even when the solutions to the system lie outside the safe set and provides ISSf-like guarantees for the invariance of an inflated safe set when the threshold ($\epsilon_1$) for invariance of the original safe set is exceeded.
In contrast, the MR-CBF does not provide asymptotic convergence guarantees when initialized outside safe set, nor does it provide ISSf-like robustness. This distinction is illustrated in the following scalar example.

\begin{figure}
    \centering
    \begin{subfigure}[t]{0.49\linewidth}
        \centering
        \includegraphics[width=\linewidth]{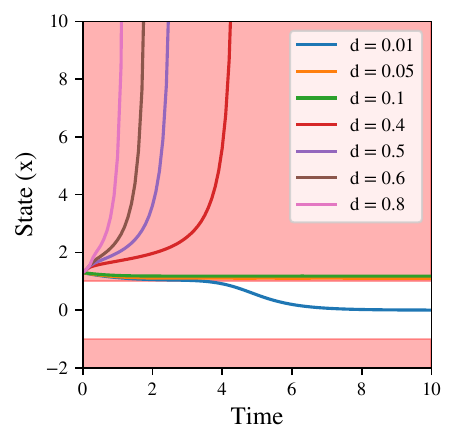}
    \end{subfigure}
    \hfill
    \begin{subfigure}[t]{0.49\linewidth}
        \centering
        \includegraphics[width=\linewidth]{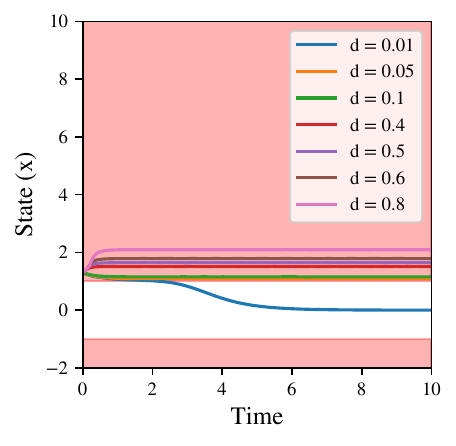}
    \end{subfigure}
    \hfill
    \caption{closed-loop trajectories when using the minimum norm controller defined by the MR-CBF (left) and the R-CBF (right) with robustness function $\rho(y) = 0.2 y(y+1)$, both constructed using (\ref{eq:counter_ex_cbf}).
    }
    \label{fig:scalar_ex}
\end{figure}

\noindent \textbf{Example 3:}
    Consider the scalar system:
    \vspace{-5pt}
    \begin{equation}
        \dot{x} = x(x-1.05)(x+1.05)+(1-x^2)u,
        \label{eq:counter_ex}
        \vspace{-5pt}
    \end{equation}
    where $x,u \in \R$, and the CBF 
    \vspace{-5pt}
    \begin{equation}
        h(x)=1-x^2.
        \label{eq:counter_ex_cbf}
        \vspace{-5pt}
    \end{equation}
    The functions $L_fh(x) = -2x^2(x-1.05)(x+1.05)$, $L_gh(x)=-2x(1-x^2)$ and $\alpha(h(x))=\alpha(1-x^2)$ are locally Lipschitz continuous and hence Lipschitz continuous on $S$.
    It can be verified that there exists a finite $d \in \R_{>0}$ such that (\ref{eq:mrcbf_delta_limit}) is satisfied for all $\delta(\hat{x}) \leq d$, and hence $h$ will also be an MR-CBF for such a $\delta$. 
    We compare the behavior of each of the closed-loop systems under the minimum norm controller ((\ref{eq:Robust_QP}) with $k_d(x)=0$) defined by the MR-CBF and the R-CBF (both constructed using (\ref{eq:counter_ex_cbf})), when initialized at the point $x=1.3$, slightly outside the safe set, for varying levels of constant state uncertainty.
    As evident from the trajectories of the closed-loop system shown in Fig. \ref{fig:scalar_ex}, beyond a certain level of state uncertainty the controller designed using the MR-CBF fails to prevent the trajectory from growing unbounded, whereas the controller designed using the R-CBF ensures that the trajectory remains within a bounded set. 

\section{Conclusion}
In this paper, we introduced the notion of an R-CBF to ensure safety in the presence of state 
uncertainty. In the absence of input constraints, R-CBFs are guaranteed to exist when regular CBFs do, however, the stronger constraints they impose lead to controllers with improved robustness properties when compared with other methodologies available in the literature.

\bibliographystyle{plain}        
\bibliography{autosam}           

@article{panagou2022safe,
  title={Safe and robust observer-controller synthesis using control barrier functions},
  author={Agrawal, Devansh R and Panagou, Dimitra},
  journal={IEEE Control Systems Letters},
  volume={7},
  pages={127--132},
  year={2022},
  publisher={IEEE}
}

@inproceedings{mrcbf2021guaranteeing,
  title={Guaranteeing safety of learned perception modules via measurement-robust control barrier functions},
  author={Dean, Sarah and Taylor, Andrew and Cosner, Ryan and Recht, Benjamin and Ames, Aaron},
  booktitle={Conference on Robot Learning},
  pages={654--670},
  year={2021},
  organization={PMLR}
}

@inproceedings{mrcbf2_iros,
  title={Measurement-robust control barrier functions: Certainty in safety with uncertainty in state},
  author={Cosner, Ryan K and Singletary, Andrew W and Taylor, Andrew J and Molnar, Tamas G and Bouman, Katherine L and Ames, Aaron D},
  booktitle={2021 IEEE/RSJ International Conference on Intelligent Robots and Systems (IROS)},
  pages={6286--6291},
  year={2021},
  organization={IEEE}
}

@inproceedings{xu2022observerCBF,
  title={Observer-based control barrier functions for safety critical systems},
  author={Wang, Yujie and Xu, Xiangru},
  booktitle={2022 American Control Conference (ACC)},
  pages={709--714},
  year={2022},
  organization={IEEE}
}

@article{jankovic2018robust,
  title={Robust control barrier functions for constrained stabilization of nonlinear systems},
  author={Jankovic, Mrdjan},
  journal={Automatica},
  volume={96},
  pages={359--367},
  year={2018},
  publisher={Elsevier}
}

@article{parameterized2023ames,
  title={Parameterized barrier functions to guarantee safety under uncertainty},
  author={Alan, Anil and Molnar, Tamas G and Ames, Aaron D and Orosz, G{\'a}bor},
  journal={IEEE Control Systems Letters},
  volume={7},
  pages={2077--2082},
  year={2023},
  publisher={IEEE}
}

@article{tissf,
  title={Safe controller synthesis with tunable input-to-state safe control barrier functions},
  author={Alan, Anil and Taylor, Andrew J and He, Chaozhe R and Orosz, G{\'a}bor and Ames, Aaron D},
  journal={IEEE Control Systems Letters},
  volume={6},
  pages={908--913},
  year={2021},
  publisher={IEEE}
}

@article{issf,
  title={Input-to-state safety with control barrier functions},
  author={Kolathaya, Shishir and Ames, Aaron D},
  journal={IEEE control systems letters},
  volume={3},
  number={1},
  pages={108--113},
  year={2018},
  publisher={IEEE}
}

@inproceedings{ersin_dist_obs,
  title={Robust safe control synthesis with disturbance observer-based control barrier functions},
  author={Da{\c{s}}, Ersin and Murray, Richard M},
  booktitle={2022 IEEE 61st Conference on Decision and Control (CDC)},
  pages={5566--5573},
  year={2022},
  organization={IEEE}
}

@article{cbf_main,
  title={Control barrier function based quadratic programs for safety critical systems},
  author={Ames, Aaron D and Xu, Xiangru and Grizzle, Jessy W and Tabuada, Paulo},
  journal={IEEE Transactions on Automatic Control},
  volume={62},
  number={8},
  pages={3861--3876},
  year={2016},
  publisher={IEEE}
}

@article{cbf_robot,
  title={Control barrier functions for mechanical systems: Theory and application to robotic grasping},
  author={Cortez, Wenceslao Shaw and Oetomo, Denny and Manzie, Chris and Choong, Peter},
  journal={IEEE Transactions on Control Systems Technology},
  volume={29},
  number={2},
  pages={530--545},
  year={2019},
  publisher={IEEE}
}

@article{cbf_aut_veh,
  title={Control barrier functions and input-to-state safety with application to automated vehicles},
  author={Alan, Anil and Taylor, Andrew J and He, Chaozhe R and Ames, Aaron D and Orosz, G{\'a}bor},
  journal={IEEE Transactions on Control Systems Technology},
  volume={31},
  number={6},
  pages={2744--2759},
  year={2023},
  publisher={IEEE}
}

@article{sector_uncert,
  title={Robust control barrier functions with sector-bounded uncertainties},
  author={Buch, Jyot and Liao, Shih-Chi and Seiler, Peter},
  journal={IEEE Control Systems Letters},
  volume={6},
  pages={1994--1999},
  year={2021},
  publisher={IEEE}
}

@article{unmodeled_dynamics,
  title={Control barrier functions with unmodeled input dynamics using integral quadratic constraints},
  author={Seiler, Peter and Jankovic, Mrdjan and Hellstrom, Erik},
  journal={IEEE Control Systems Letters},
  volume={6},
  pages={1664--1669},
  year={2021},
  publisher={IEEE}
}

@inproceedings{xu2023disturbance,
  title={Disturbance observer-based robust control barrier functions},
  author={Wang, Yujie and Xu, Xiangru},
  booktitle={2023 American Control Conference (ACC)},
  pages={3681--3687},
  year={2023},
  organization={IEEE}
}

@article{cbf_robustness,
  title={Robustness of control barrier functions for safety critical control},
  author={Xu, Xiangru and Tabuada, Paulo and Grizzle, Jessy W and Ames, Aaron D},
  journal={IFAC-PapersOnLine},
  volume={48},
  number={27},
  pages={54--61},
  year={2015},
  publisher={Elsevier}
}

@article{compendium_class_k,
  title={A compendium of comparison function results},
  author={Kellett, Christopher M},
  journal={Mathematics of Control, Signals, and Systems},
  volume={26},
  pages={339--374},
  year={2014},
  publisher={Springer}
}

@article{dynamic_env_segway,
  title={Safety-critical control with input delay in dynamic environment},
  author={Molnar, Tamas G and Kiss, Adam K and Ames, Aaron D and Orosz, G{\'a}bor},
  journal={IEEE transactions on control systems technology},
  volume={31},
  number={4},
  pages={1507--1520},
  year={2022},
  publisher={IEEE}
}

@book{convex_book,
  title={Convex optimization},
  author={Boyd, Stephen P and Vandenberghe, Lieven},
  year={2004},
  publisher={Cambridge university press}
}

@article{cbf_satellite,
  title={Robust control barrier functions under high relative degree and input constraints for satellite trajectories},
  author={Breeden, Joseph and Panagou, Dimitra},
  journal={Automatica},
  volume={155},
  pages={111109},
  year={2023},
  publisher={Elsevier}
}

@inproceedings{nguyen2016exponential,
  title={Exponential control barrier functions for enforcing high relative-degree safety-critical constraints},
  author={Nguyen, Quan and Sreenath, Koushil},
  booktitle={2016 American Control Conference (ACC)},
  pages={322--328},
  year={2016},
  organization={IEEE}
}

@article{alyaseen2025continuityQP,
  title={Continuity and boundedness of minimum-norm CBF-safe controllers},
  author={Alyaseen, Mohammed and Atanasov, Nikolay and Cort{\'e}s, Jorge},
  journal={IEEE Transactions on Automatic Control},
  volume={70},
  number={6},
  pages={4148--4154},
  year={2025},
  publisher={IEEE}
}

\end{document}